\documentclass[a4paper]{article}

\usepackage{fullpage,url,amsmath,amsfonts,amssymb, tedmath}
\newcommand{\tr}{\mbox{tr}}
\newcommand{\te}{\theta}
\newcommand{\cc}{\mathbb{C}}

\title{Full diversity sets of unitary matrices from orthogonal sets  
of idempotents  }
\author{
Ted Hurley\footnote{National Universiy of Ireland Galway; email:
Ted.Hurley@NuiGalway.ie }}
\date{}

\begin{document}

\maketitle

\begin{abstract} Orthogonal sets of idempotents are used to
  design sets  of unitary matrices, known as  constellations, such that 
 the modulus of the determinant of the
 difference of any two distinct elements is greater than $0$.  
It is shown that unitary matrices in general are derived 
from orthogonal sets of idempotents reducing  
the design problem to  a construction problem of unitary matrices 
from  such sets.  The 
 quality of the constellations  constructed in this way  
and the actual differences
 between the unitary matrices can be determined algebraically
 from the idempotents used. 
This has
 applications to the design of unitary space time constellations.
\end{abstract}
\section{Introduction}
The design problem for unitary space time
constellations is set out nicely in \cite{mult} and \cite{orig}: ``Let $M$ be
the number of transmitter antennas and $R$ 
the desired transmission rate. Construct a set $\mathcal{V}$ of $L =
2^{RM}$ unitary $M\times M$ matrices such that for any two distinct elements $A,B$
in $\mathcal{V}$, the quantity $|\det(A-B)|$ is as large as possible. Any
set $\mathcal{V}$ such that $|\det(A-B)|> 0$ for all distinct $A,B$ is
said to have {\em full diversity}.''

The number of transmitter antennas is the size $M$ of the matrices and
this is also known as the order of the constellation or matrices. 
`Order' in this instance refers to the size of the matrices.

The set $\mathcal{V}$ is known as a {\em constellation}. In
\cite{mult} also it is explained that the {\em quality} of the
constellation is measured by

 $$\zeta_{\mathcal{V}} = \frac{1}{2} \min_{V_l, V_m \in \mathcal{V}, V_l \neq
    V_m} |\det (V_l-V_m)|^{\frac{1}{M}}$$

Here we present  general methods for constructing such
constellations  from  orthogonal sets 
 of idempotents. It is shown that unitary matrices are  obtained
 from complete orthogonal sets of 
 idempotents in a precise manner. 
 This enables constructions of constellations using such
 representations and the nature
 of the constructions allows the    
 quality to be determined algebraically;     
 all differences may often be explicitly calculated. 

New constellation are derived from the general concept, explict
constructions are given  and
many more may be derived. Indeed 
infinite series of fully diverse real and infinite series 
of fully diverse complex constellations may be constructed using the
methods; from these finite sets may be chosen and the quality worked out
algebraically as required. 

Extension methods for constructing  constellations are derived. 
Algebraic results on differences of unitary matrices are formulated 
which may then be used  
to calculate the quality of such constructed constellations. 

A method is derived in Section \ref{tangle} 
which allows the construction
of constellations of order  $2n\ti 2n$ from constellations of order $n\ti
n$ where the higher order constellations have similar quality and
similar rate to the lower order constellations. In this way many more
constellations of higher order may be constructed from those already
constructed.

Examples are constructed which show some of the range of the methods but
the methods are fairly general and many more may be constructed. 

{\em Division algebras} have also been used in this area and 
the excellent survey article \cite{ams} and the references therein
give the details. 
See also \cite{berhuy}, \cite{channel}.

\subsection{Further notation}
For unitary matrices $A,B$ of the same $M\ti M$ size  define the {\em distance}
or {\em difference} between $A$ and $B$ to be $\frac{1}{2}|\det(A-B)|^{\frac{1}{M}}$. Thus
for a constellation $\mathcal{V}$ of unitary matrices  consisting of $M\ti M$
matrices  its quality is the  
minimum of the distances between any two different matrices in
$\mathcal{V}$. 

A  set of orthogonal idempotents in a ring $R$ is a set 
 $\{e_1, e_2, \ldots, e_k\}$ satisfying: \\
(i) $e_i \not = 0$ and $e_i^2 = e_i$, $1\leq i\leq k$.\\ (ii) If
$i\not = j$ then $e_ie_j = 0$. 
\\ If further $1 = e_1+e_2 + \ldots + e_k$ then the set is said to
be {\em a complete set of orthogonal idempotents}.

Here we use $1$ for the identity of $R$. In general $1$
will denote the identity of the system under consideration. 

The idempotent $e_i$ is said to be {\em primitive} if it cannot be
written as $e_i= e_i^{'}+ e_i^{''}$ where $e_i^{'},e_i^{''}$ are idempotents
such that $e_i^{'},e_i^{''} \neq 0$ and $e_i^{'}e_i^{''}=0$. A 
set of orthogonal idempotents is said to be {\em primitive} if each
idempotent in the set is primitive.

A mapping $^*: R\to R$ in which $r\mapsto r^*, (r\in R)$ 
is said to be an {\em involution} on $R$ if and only if (i) $r^{**} =
r, \, \forall r\in R$, (ii) $(a+b)^* =
a^*+b^*, \, \forall a,b \in R$, and 
(iii) $(ab)^* = b^*a^* , \, \forall a,b \in R$. 

We are particularly interested in the
case where $^*$ denotes complex conjugate transpose in the case of
matrices over $\cc$ and denotes transpose for matrices over other
fields and in particular over $\mathbb{R}$, the reals. 

If $R$ has an involution $^*$ then an element $v\in R$ is said to be
{\em symmetric} (with respect to $^*$) if $v^*=v$  
and a set of elements is said to be symmetric if each element
in the set is symmetric.
 
The matrix $U\in R_{n\ti n}$ is said to be a unitary matrix (with
respect to $^*$) if $UU^*=1$. 

A constellation is said
to be {\em fully diverse} when it has full diversity.

Further general algebra background may be found \cite{blahut} although
little background in coding theory itself is required. 

\subsection{Layout} In Section \ref{unitary} the connection between
orthogonal sets of idempotents and unitary matrices is established and
in the (sub)Section \ref{idems} properties of,  and construction
methods for, orthogonal sets of idempotent matrices are analysed. 

In Section \ref{complete} methods are
derived for constructing unitary matrices from a complete set of
idempotents. The results here correspond to those obtained in the
cyclic case as in \cite{mult} and \cite{orig}. Examples are given and
the rates and quality are worked out. 

Methods are derived for constructing
and analysing unitary matrices using different sets of
orthogonal idempotents in Section \ref{different}.  In (sub)Sections 
\ref{first1}, \ref{first2}, \ref{sincos} the methods are applied to
constructing infinite fully diverse  
sets of real unitary  and their distances are
established; from such sets finite subsets may be taken as required 
and the 
distances and quality of these sets are then known. Examples are given here
and many more may be deduced. 
In Section \ref{complex} the methods are applied to 
constructing sets of constellations with complex entries; indeed infinite such sets
are constructed from which finite subsets may be deduced as required.
 
In Section
\ref{tangle}, a method, using what is called a tangle of matrices, is
devised  to construct sets of $2n\ti 2n$ fully diverse constellations from a set
of $n\ti n$ fully diverse constellations. The quality of the $2n\ti 2n$
constellations may be given in terms of the quality of the
$n\ti n$ constellations from which they are derived. 

\subsection{Dependence} Some of the sections may be read independently
except where a reference is made  to an example constructed
in a previous
section. In this sense Sections \ref{complete}, \ref{different},
\ref{tangle} may be read independently. The (sub)Sections \ref{idems} and
\ref{grmat1}, on methods
for constructing
orthogonal sets of idempotent matrices and properties therefrom,  
may be consulted as required.

\subsection{Determinants of block matrices} Interested will be in
$P=\det( \begin{pmatrix} A & B \\ C & D \end{pmatrix})$ where $A,B,C,D$
are block matrices of the same size. It is not necessary that all of
$A,B,C,D$ commute in order to have a formula (such as below) 
 for $P$ in terms of $A,B,C,D$. 

Let $M=\begin{pmatrix} A & B \\ C & D \end{pmatrix}$.

Then \begin{enumerate}

\item $\det M = \det (AD-BC)$ whenever at least one of $A,B,C,D$ is
  equal to the zero matrix.
\item  $\det M = \det (AD-BC)$ when $DC=CD$.
\item  $\det M = \det (AD-CB)$ when $AC=CA$.
\item $\det M = \det (DA-BC)$ when $BD=DB$.
\item $\det M = \det (DA-CB)$ when $AB=BA$.
\end{enumerate}
Such
results may be found on-line or in for example \cite{commblock}. These
will be applied without further reference.   
  
\section{Unitary matrices and orthogonal sets of idempotents}\label{unitary}
Unitary matrices over $\cc$ are built from complete symmetric
orthogonal sets of matrices as follows: 
\begin{proposition}\label{7}  $U$ is a unitary $n\ti n$ matrix over $\cc$ if and
 only if $U = \al_1 v_1^*v_1 + \al_2 v_2^*v_1 + \ldots + \al_nv_n^*v_n$
 where $\{v_1^*v_1, v_2^*v_2, \ldots, v_n^*v_n\}$ is a complete
 symmetric orthogonal set of idempotents in $\cc_{n\ti n}$
 and $\al_i\in \cc$ with $|\al_i|=1$, $ \forall i$. Further the $\al_i$ are
 the eigenvalues of $U$.
\end{proposition}

This result appears in \cite{hur5,hur6} but as it leads to fundamental
constructions, a proof
is given here for completeness.

\begin{proof} Let $U= \al_1 v_1^*v_1 + \al_2 v_2^*v_2 + \ldots +
 \al_nv_n^*v_n$ where $\{v_1^*v_1, v_2^*v_2, \ldots, v_n^*v_n\}$ is a
 orthogonal complete set of idempotents 
 with $|\al_i| =1$. It is easy to check that $UU^*= 1$. Then $Uv_i^*=\al_iv_i^*$ and so the $\al_i$ are
 the eigenvalues of $U$.  

Suppose then $U$ is a unitary
 matrix.  It is known, as in particular $U$ is a normal matrix, 
 that there exists a unitary matrix $P$ such that 
$U=P^*DP$ where $D$ is diagonal and  the entries of $D$ must have 
modulus $1$. Thus $P=
\begin{ssmatrix} v_1 \\ v_2 \\ \vdots \\ v_n \end{ssmatrix}$ where $\{v_1,
 v_2, \ldots , v_n\}$ is an orthonormal basis (of row vectors) for $\cc_n$
 and $D=\diag(\al_1,\al_2, \ldots, \al_n)$ with $|\al_i| =1$ and the
 $\al_i$ are the eigenvalues of $U$.
Then \begin{eqnarray*}  U = P^*DP  \\ =& (v_1^*, v_2^*, \ldots,
 v_n^*)\begin{ssmatrix} \al_1 & 0 & \ldots & 0\\ 0& \al_2 & \ldots & 0
\\ \vdots & \vdots & \vdots & \vdots \\ 0 & 0 & \ldots & \al_n \end{ssmatrix} 
 \begin{ssmatrix} v_1 \\ v_2 \\ \vdots \\ v_n \end{ssmatrix} \\ =& 
(\al_1v_1^*, \al_2v_2^*, \ldots, \al_nv_n^*)\begin{ssmatrix} v_1 \\ v_2
					    \\ \vdots \\ v_n
					    \end{ssmatrix}
\\  = &\al_1v_1^*v_1 + \al_2v_2^*v_2 + \ldots + \al_nv_n^*v_n.
\end{eqnarray*}
 
\end{proof}

Thus unitary matrices are generated by complete symmetric orthogonal
sets of idempotents formed from the diagonalising unitary matrix. Notice
that the $\al_i$ are the eigenvalues of $U$.

\subsection{Example}
For example consider the real orthogonal/unitary matrix
$U=\begin{ssmatrix} \cos \theta & \sin \theta \\ -\sin \theta & \cos
    \theta \end{ssmatrix}$.
This has eigenvalues $e^{i\theta}, e^{-i\theta}$ and
$P=\frac{1}{\sqrt{2}}\begin{ssmatrix} -1 &-i \\ i& 1 \end{ssmatrix}$ is
a diagonalising unitary matrix. Take the rows $v_1=
\frac{1}{\sqrt{2}}(-1,-i), \, v_2=\frac{1}{\sqrt{2}}(i,1)$
of $P$ and consider the complete orthogonal symmetric set of idempotents
$\{P_1 = v_1^*v_1 = \frac{1}{2}\begin{ssmatrix}1 & -i \\ i & 1
\end{ssmatrix}, P_2 = v_2^*v_2 = \frac{1}{2}\begin{ssmatrix}1 & i \\ -i & 1
\end{ssmatrix} \}$.

Then applying Proposition \ref{7} gives  $U=e^{i\theta}P_1+ e^{-i\theta}P_2 =
\frac{1}{2}e^{i\theta}\begin{ssmatrix}1 & -i \\ i & 1 
		      \end{ssmatrix} + \frac{1}{2}e^{-i\theta}
\begin{ssmatrix}1 & i \\ -i & 1
\end{ssmatrix} $, which may be checked independently to be equal to $U$.  
\subsection{Complete orthogonal sets of idempotents}\label{idems}
Unitary matrices are designed from complete
symmetric sets of orthogonal idempotents as in Proposition \ref{7}.

Also in \cite{hur5, hur6}  construction methods for 
 complete symmetric orthogonal systems of idempotents are given. 
The methods  are based 
 essentially on (a)
orthogonal projections; (b) group rings. 
The reader may consult the results in these papers  as required later. 
 Methods similar have been used to
construct series of {\em paraunitary matrices}
which play an important
role in signal processing, \cite{hur5}.

\subsection{Rank and Determinants}\label{grmat1}

The results in this subsection are used later for constructing
constellations and for calculations the differences and quality. 
They appear essentially in \cite{hur5,hur6} but in a slightly
different form. 

\begin{lemma}\label{trrank} Suppose $\{E_1,E_2, \ldots, E_s\}$ is a
set of orthogonal idempotent matrices. Then $\rank
(E_1+E_2 +\ldots + E_s) = \tr (E_1+E_2+ \ldots + E_s) = \tr E_1+ \tr
 E_2+ \ldots + \tr E_s = \rank E_1+ \rank E_2 + \ldots +\rank
E_s$.
\end{lemma}
\begin{proof}
It is known that $\rank A = \tr A$ for an idempotent matrix, see
for example \cite{idemrank}, and so
$\rank E_i = \tr E_i$ for each $i$. If $\{E,F,G\}$ is a set an orthogonal
 idempotent matrices so is  $\{E+F,G\}$. From this it follows that $\rank
(E_1+E_2 +\ldots + E_s) = \tr (E_1+E_2+ \ldots E_s)= \tr E_1+\tr E_2 +
 \ldots + \tr E_s = \rank E_1+ \rank E_2 + \ldots \rank
E_s$.
\end{proof}
\begin{corollary}\label{trrank1}
$\rank(E_{i_1}+ E_{i_2}+ \ldots + E_{i_k})= 
\rank E_{i_1} +\rank E_{i_2}+ \ldots + \rank E_{i_k}$ for $i_j \in \{
1,2,\ldots, s\}$, $i_j\neq i_l$.
\end{corollary}

Let $\{e_1, e_2, \ldots, e_k\}$ be a complete orthogonal set of idempotents
in a vector space over $F$. 
\begin{theorem}\label{gr1} Let $w= \al_1 e_1 + \al_2 e_2 + \ldots +
\al_ke_k$ with $\al_i \in F$. Then $w$
 is invertible if and only if each $\al_i \neq 0$ 
and in this case $w^{-1}
 = \frac{1}{\al_1}e_1 + \frac{1}{\al_2}e_2+ \ldots + \frac{1}{\al_k}e_k$.
\end{theorem} 

\begin{proof} Suppose each $\al_i \neq 0$.
Then $w(\frac{1}{\al_0}e_0+\frac{1}{\al_1}e_1+ \ldots + \frac{1}{\al_k}e_k)
 = e_0^2 + e_1^2 + \ldots + e_k^2 = e_0+e_1+\ldots + e_k = 1$.

Suppose $w$ is invertible and that some $\al_i=0$. 
Then $we_i =0$ and so $w$ is a (non-zero) zero-divisor and is not invertible.
\end{proof}

We now specialise the $e_i$ to be $n\ti n$ matrices and in this case
use capital letters and let $e_i = E_i$.
 
Let $A= a_1 E_1 + a_2 E_2 + \ldots + a_kE_k$. Then $A$
 is invertible if and only if each $a_i \neq 0$ and in this case $A^{-1}
 = \frac{1}{a_1}E_1 + \frac{1}{a_2}E_2+ \ldots + \frac{1}{a_k}E_k$. 

The following result is very useful for determining the quality of
constellations constructed by the methods of idempotents. 
\begin{theorem}{\label{det}} Suppose $\{E_1, E_2, \ldots, E_k\}$ is a
 complete symmetric orthogonal set of idempotents in $F_{n\ti n}$. 
Let $A= a_1 E_1 + a_2 E_2 + \ldots +
 a_kE_k$.  Then the determinant of $A$ is 
$|A| = a_1^{\rank E_1}a_2^{\rank E_2}\ldots a_k^{\rank E_k}$.
\end{theorem}

\begin{proof} Now $AE_i = a_iE_i^2 = a_iE_i$. Thus each column of $E_i$ is
 an eigenvector of $A$ corresponding to the eigenvalue $a_i$. Thus there
 are at exist $\rank E_i$ linearly independent eigenvectors corresponding to
 the eigenvalue $a_i$. Since $\rank E_1 + \rank E_2 + \ldots + \rank
 E_k = n$ there are exactly $\rank E_i$ linearly independent
 eigenvectors corresponding to the eigenvalue $a_i$. 
Let $r_i =\rank E_i$. Let these $r_i$
linearly independent eigenvectors corresponding to $a_i$ be
 denoted by $v_{i,1}, v_{i,2}, \ldots v_{i,r_i}$. Do this for each $i$.

Any  column of $E_i$ is perpendicular to any column of $E_j$ for
$i\neq j$ as  $E_iE_j^* = 0$.

Suppose now $\sum_{j=1}^{r_1} \alpha_{1,j} v_{1,r_j} +
 \sum_{j=1}^{r_2}\alpha_{2,j}v_{2,r_j} + \ldots +
 \sum_{j=1}^{r_k}\alpha_{k,j}= 0$.

Multiply through by $E_s$ for $1\leq s \leq k$. This gives
 $\sum_{j=1}^{r_k}\alpha_{k,j}v_{k,j} = 0$ from which it follows that
 $\alpha_{k,j} = 0$ for $j=1,2, \ldots r_k$.  

Thus the set of vectors $S=\{v_{1,1}, v_{1,2}, \ldots v_{1,r_1},
 v_{2,1}, v_{2,2},
 \ldots, v_{2,r_2} \ldots, \ldots, v_{k,1}, v_{k,2}, \ldots, v_{k,r_k}\}$
 is linearly independent and form a basis for $F^n$ -- remember that
 $\rank (E_1+E_2+\ldots + E_k)= n$. Hence $A$ can be
 diagonalised by the matrix of these vectors and thus there is a
 non-singular matrix $P$ such that 
$P^{-1}AP = D$ where $D$ is a diagonal matrix consisting of 
 the $a_i$ repeated $r_i$ times for each $i=1,2, \ldots k$.

Hence $|A| = |D| = a_1^{r_1}a_2^{r_2}\ldots a_k^{r_k}$.
\end{proof}

Theorem \ref{det} may 
be used to compute the full distribution of the differences in a
constellation in certain cases.    

\vspace{.2in}

The  following Proposition may be found in \cite{hur5},
Proposition 4.2. 
\begin{proposition}\label{idem1} Let $F$ be a field in which every element has a square
 root. Suppose also an involution $^*$ is defined on the set of
 matrices over $F$. Then $P$ is a symmetric (with respect to $^*$) 
 idempotent of $\rank 1$ in $F_{n\ti n}$
if and only if $P= vv^*$ where $v$ is a column
 vector such that $v^* v=1$.  
\end{proposition}

(Note that `symmetric with respect to $^*$' in the case of
matrices over $\cc$ is usually  termed `Hermitian'.)

It is necessary that square roots exist in the field and an example is
given in \cite{hur5} to demonstrate this. 

Proposition \ref{idem1} shows  that $vv^*$,
for $v$ a unit column vector,
is a symmetric idempotent and sets of unitary matrices are constructed from
these types of idempotents in later sections.  

\section{Constellations from complete orthogonal 
set of idempotents}\label{complete}

Recall that a set of unitary matrices is said to have {\em full
diversity} or to be fully diverse if and only 
if the modulus of the determinant of the difference of any  two
matrices in the set is non-zero. 
\begin{theorem}\label{diversity} Let $\{E_1,E_2, \ldots, E_k\}$ be a complete
  symmetric orthogonal set of idempotents.
Define for $s=1,2,\ldots, t$, $V_{s}= \sum_{j=1}^k \al_{s,j} E_j$
where the $\al_{p,q}$ are complex numbers of modulus $1$. Then
$\mathcal{V} = \{V_1,V_2, \ldots, V_t\}$ is a constellation of unitary
  matrices. Further $\mathcal{V}$ has full diversity if and only if
  for each $t=1,2,\ldots, k$, 
$\al_{r,t}\neq \al_{s,t}$ when $r\neq s$. 
\end{theorem}
\begin{proof} It is easy to check that each $V_j$ is a unitary matrix.
Consider all $W(p,l) = V_p - V_l$ for $p\neq l$. Then 
$W(p,l) = \sum_{i=1}^k (\al_{p,i}-\al_{l,i})E_i$ and  by Theorem
\ref{det},  $\det(W(p,l)) =  (\al_{p,1}-\al_{l,1})^{\rank
  E_1}(\al_{p,2}-\al_{l,2})^{\rank E_2}\ldots
(\al_{p,k}-\al_{l,k})^{\rank E_k}$.
Hence $|\det(W(p,l))| \neq 0$ for all $p,l$ with $ p\neq l$ 
if and only if $\al_{p,i}\neq \al_{l,i}$ for $1\leq i \leq k$ and for
all $p,l$ with $p\neq l$. 
\end{proof} 

Theorem \ref{diversity} enables
the construction of classes of constellations, and 
Theorem \ref{det} enables the calculation of the
quality of each one and indeed the calculation of all  the differences
in the constellation. 
\vspace{.1in}

\subsection{Examples}
Theorem \ref{diversity} is now  used to construct constellations
and  calculate their qualities. 
To keep full diversity it is only
necessary to adhere to the conditions of Theorem \ref{diversity}.

Consider the $2\ti 2$ case. 
Let $\{E_0,E_1\}$ be a complete orthogonal set of idempotents in
$\cc_{2\ti 2}$.

For fixed $\te$, use $(k,j)$ to mean
$e^{ki\te}E_0+e^{ji\te}E_1$. Use $|(p,q)|$ to mean
$|1-e^{ip\te}||1-e^{iq\te}|$. Note that $|1-e^{i\al}| =
|1-e^{-i\al}|$. 
 
\begin{enumerate} \item 
$n=5, \, \te = 2\pi/5$:
$\mathcal{V}=\{(0,2),(1,4),(2,1),(3,3),(4,0)\}$. Then quality is
$(\sin(\pi/5)\sin(2\pi/5))^{1/2}= 0.74767...$ . The rate here is
  $\log_2(5)/2 =  1.1609..$.
\item 
$n=8$: $\mathcal{V}=
\{(0,0),(1,3),(2,6),(3,1),(4,4),(5,7),(6,2),(7,5)\}$. This is the form
$\{(j,3j)\, | \, j= 0, 1,\ldots, 7\}$ where $3j$ is interpreted as  $ 3j\mod 8$.
Here the modulus of  difference
between any pairs of these is at worst $|(2,2)|$ or $|(1,3)|$. It is
easy to see that $|(2,2)|> |(1,3)|$. Then quality is 
$(\sin(\pi/8)\sin(3\pi/8))^{1/2}=0.5946...$. The rate is $\log_2(8)/2
= 1.5$. 

It is possible here also  to determine the distribution of the
differences.
There are in total $28$
differences between the $8$ elements of the constellations. These are
distributed as follows. 
There are $16$ of modulus
$|\sin(\pi/8)\sin(3\pi/8))^{1/2}|= 0.5946..$, $8$ of modulus  
$\sin(2\pi/8)\sin(2\pi/8)0^{1/2} = 0.707$ and  $4$ of modulus
$(\sin(4\pi/)\sin(4\pi/8))^{1/2}=1$. A weighted average
is  $0.684657..$, which may
possibly be a more correct measure of a `quality' subject to the
main quality.
\item $n= 32$. Let the constellation be  
$\{(j,7j\mod 32) \, | \, j=0,1,\ldots, 31\}$. On noting that
$|(5,3)| > |(1,7)|$ it is seen that the quality of this is
$\frac{1}{2}(|1-e^{i\te}||1-e^{7i\te}|)^{1/2}
= (\sin(\pi/32)\sin(7\pi/32))^{1/2}$ and this  is $0.24936 ..$.
Here also the distribution of the differences may be determined.
\item $n=64$. Use the constellation $\{(i,19i) \, | \, i= 0, 1, \ldots,
  63\}$ with $19i$ interpreted as $19i \mod 64$. Then get quality
  $(\sin(\pi/64)\sin(19\pi/64))^{1/2} = 0.1985 ..0$. Note here that
  $|(10,2)| > |(19,1))$. The rate here is $\log_2(64)/2 = 3$. 
\item $n=128$. Use the constellation $\{(i,47i) \, | \, i=0,1, \ldots,
  127\}$. Then get quality $(\sin(\pi/128)\sin(47\pi/128))^{1/2}=
  0.14978..$. There are other pairs $(i,47i)$ where $47i \mod 128 < 47$
  but here the $i$ is big enough so that $|(i,47i)| > |(1,47)|$. The
  rate here is $\log_2(128)/2 = 3.5$. 
\end{enumerate}   
\subsection{Higher order}
Let $\{E_0,E_1,E_2\}$ be a complete orthogonal set of idempotents in
$C_{3\ti 3}$ Use $(k,l,m)$ to mean $ 
e^{ki\te}E_0+e^{li\te}E_1+e^{mi\te}E_2$ when $\te$ has been given a value.
 
Suppose now $n=8$ and $e^{i\te} = e^{2i\pi/8}$ is a primitive $8^{th}$
root of $1$. Consider the constellation:

$\{(0,0,7), (1,3,2),(2,6,5), (3,1,3),(4,4,0), (5,7,1), (6,2,6),(7,5,4)\}$.

This has quality $(\sin(\pi/8)\sin(3\pi/8)\sin(\pi/8))^{1/3}=
0.51337..$. There are just two of the $28$ differences with this least
modulus. All the other differences have modulus at least
\\ $(\sin(2\pi/8)\sin(2\pi/8)\sin(\pi/8))^{1/3} = 0.5762 ...$ .

For even order (i.e. for $\cc_{2k\ti2k}$) 
 by simply repeating the $2$
distribution it is possible to obtain the same quality as the $2\ti 2$
case but with smaller rate. Improvements on this can also be
obtained by matching `bad' pairs with `good' pairs.

For example consider a complete orthogonal set of idempotents
$\{E_1,E_2,E_3,E_4\}$  in $\cc_{4\ti 4}$. Then 
repeat the pattern for the
$2\ti 2$ case. 
Let $\te = 2\pi/8$ and  
define $(j,k,l,m)= e^{ji\te}E_1+
e^{ki\te}E_2+e^{li\te}E_3+e^{li\te}E_4$.
Define the constellation $\{(j,3j,j,3j) \, | \, j= 0,
1,2,\ldots, 7\}$ where $3j$ is interpreted as $ 3j \mod 8$.
Then the quality of this constellation is
also $(\sin^2(\pi/8)\sin^2(3\pi/8))^{1/4}=0.5946...$. 
With modification to the construction the quality can be 
improved to $(\sin^2(\pi/8)\sin(3\pi/8)\sin(4\pi/8))^{1/4}=0.60649 ...$.

\subsection{General construction} 
In general consider a complete orthogonal set of idempotents
$\{E_1,E_2, \ldots, E_s\}$ in $\cc_{n\ti n}$. Let 
$e^{2\pi/k}= e^{i\te}$ be a primitive $k^{th}$ root of $1$.
Define $(j_1,j_2,\ldots, j_k) = \di\sum_{t=1}^se^{j_ti\te}E_t$. 
Then construct the constellation $\{(j_{i1},j_{i2},\ldots, j_{is}) \, |
\, i=1,2,\ldots, r, 1\leq r \leq k\}$ where, for each $i$,  
$j_{ti}\neq j_{qi}$ for $t\neq q$.
Maximise the quality by  choosing the $j_{kl}$ to maximise the minimum
modulus of the differences. 

This could be further developed but is left for consideration
elsewhere. 


\section{Constellations combining different orthogonal sets of
  idempotents}\label{different} 

Now consider constructing constellations by combining  different sets of
orthogonal idempotents. 
Good constellations and indeed good constellations 
 with real unitary matrices and good quality can still be 
be obtained.

\subsection{Symmetric  $2\ti 2$ idempotents}\label{first1}
Let $\{P,P_1\}$ be a complete symmetric 
orthogonal set of idempotents in $\cc_{2\ti 2}$. Then
$P_1=1 - P$ and since $P$ is a symmetric idempotent 
 it follows that $P=v^*v$ for a unit row-vector $v$ by Proposition 
\ref{idem1}. Let $v= \begin{pmatrix} a \\ b \end{pmatrix}$ 
and thus 
$P= vv^*=\begin{pmatrix}|a|^2 &ab^* \\ ba^* & |b|^2
\end{pmatrix}$ where $|a|^2+|b|^2=1$.

Consider also the idempotent $Q=\begin{pmatrix}|c|^2
  &cd^* \\ dc^* & |d|^2 
\end{pmatrix}$ where $|c|^2+|d|^2=1$.

\begin{proposition}\label{result} Let $P=\begin{pmatrix} a \\ b \end{pmatrix}
 (a^*,b^*)$ and $Q=\begin{pmatrix} c\\ d \end{pmatrix} (c^*,d^*)$. 
Then $\begin{pmatrix} a & c \\ b& d \end{pmatrix}\begin{pmatrix} a^* &
						  b^* \\ -c^* & -d^*
						 \end{pmatrix} = P-Q$.
\end{proposition}

\begin{proof} This may be shown directly by matrix multiplication.
\end{proof}
\begin{corollary}\label{result1} $|\det(P-Q)|= |\det(\begin{pmatrix} a & c \\ b& d
\end{pmatrix}\begin{pmatrix} a^* & b^* \\ -c^* & -d^* \end{pmatrix}|= 
|\det(\begin{pmatrix} a & c \\ b& d \end{pmatrix})^2|$
\end{corollary}
\begin{proof} This follows since  $|\det X |= |\det(-X)| = |\det
 X^*|$ for a matrix $X$.
\end{proof}
\begin{corollary}\label{detreal1}  $|\det(2(P-Q))| =
  4|ad-bc|^2$.
\end{corollary}
\begin{corollary} $|\det(P-Q)| = 0 $ if and only if $ad=bc$. 
\end{corollary}
\begin{corollary} When $a,b,c,d$ are real,  $|\det(2(P-Q))| =
  4(ad-bc)^2$.
\end{corollary}

\quad The following more general result is needed later. 
\begin{lemma}\label{gone} $\begin{pmatrix} aa^*-\al cc^* & ab^* -\al cd^* \\ ba^* -\al dc^* & bb^* - \al dd^* \end{pmatrix} = \begin{pmatrix} a & \sqrt{\al}c \\ b & \sqrt{\al}d \end{pmatrix}\begin{pmatrix} a^* & b^* \\ -\sqrt{\al}c^* & -\sqrt{\al}d^* \end{pmatrix}.$
\end{lemma}
\begin{proof} This follows by direct matrix multiplication.
\end{proof} 
\begin{corollary} Suppose $E=\begin{pmatrix} aa^* & ab^* \\ ba^* & bb^* \end{pmatrix}, F=\begin{pmatrix} cc^* & cd^* \\ dc^* & dd^* \end{pmatrix}$. Then 
$|\det(E-F)| = |\det(\begin{pmatrix} a & c \\ b & d \end{pmatrix})^2|$ and $|\det (E-\al F)| = |\al\det(E-F)|$.  
\end{corollary}

\begin{proof} That $|\det(E-F)| = |\det(\begin{pmatrix} a & c \\ b & d \end{pmatrix})^2|$ follows from Corollary \ref{result1}. 

Now $|\det(E-\al F)| = |\det(\begin{pmatrix} a & \sqrt{\al}c \\ b & \sqrt{\al}d \end{pmatrix}\begin{pmatrix} a^* & b^* \\ -\sqrt{\al}c^* & -\sqrt{\al}d^* \end{pmatrix}|= |\sqrt{\al}\sqrt{\al}\det( \begin{pmatrix} a & c \\ b & d \end{pmatrix}\begin{pmatrix} a^* & b^* \\ -c^* & -d^* \end{pmatrix})| = |\al \det(E-F)|$.

\end{proof}

Now $U=P-P_1= 2P-I$ and $V=Q-Q_1=2Q-1$ are unitary matrices
and $U-V=2(P-Q)$. Note that $U,V$ do not in general commute so
cannot be simultaneously diagonalised. Thus these constructions are not
the same as the cyclic constructions. 

More generally we have the following result. 
\begin{lemma}\label{unitary1} Let  $E$ be a  symmetric idempotent
  in $\cc_{n\ti n}$. Then
 $2E-I_n$ is a unitary matrix.
\end{lemma}
\begin{proof} Write $I$ for $I_n$. Then  
$(2E-I)(2E-I)^* = (2E-I)(2E-I)= 4E^2 - 2E-2E + I= 4E-4E+I=I$.
\end{proof}

Now form constellations of the form $\{2E_i -I | i\in J\}$ for some
index set $J$ where $E_i$ are symmetric idempotents. These matrices do
not commute in general so such constellations are certainly different
to those in Section  \ref{complete} and are for example cyclic.

Suppose the matrices are of order $2\ti 2$. 
The differences of the elements in the constellation are
$\frac{1}{2}|(\det(2E_i-I-(2E_j-I)|^{1/2}) = \frac{1}{2}|\det 2 (E_i-E_j)|^{1/2}=
\frac{1}{2}|4\det (E_i-E_j)|^{1/2} = |\det(E_i-E_j)|^{1/2}$. These can be
calculated by results such as 
Corollary \ref{detreal1} for certain idempotents, for example for those
of the form  $E_i=v_i^*v_i$, for unit row vectors $v_i$. 
In constructing constellations, using unitary matrices of the form 
$2E_i-I$, it is desirable  to make the modulus of the determinants of
the difference between any two $E_i$ and $E_j$ as large as
possible. 
\subsection{Sets of real unitary constellations}\label{first2} 

Consider the following idempotents formed using $vv^*$
for a row vector $v$. Assume now the entries of $v$ are real. 
Each matrix is of the form  
$vv^*=\begin{pmatrix}|a|^2 &ab \\ ba & |b|^2
\end{pmatrix}$ where $|a|^2+|b|^2=1$.

Let $v_1= \begin{pmatrix} \frac{1}{\sqrt{2}} \\ \frac{1}{\sqrt{2}} \end{pmatrix}$
to give $E_1 = \begin{pmatrix}1/2 & 1/2 \\ 1/2 & 1/2
\end{pmatrix}$. This gives the unitary real matrix $A_1=2E_1-I$.  
Now let $a_2^2=1/3, b_2^2=2/3$ to give $E_2 = \begin{pmatrix}1/3 &
  \sqrt{2}/3 \\ \sqrt{2}/3 & 2/3\end{pmatrix}$. This gives the unitary
  real unitary matrix $A_2= 2E_2-I$.

In general let  $a_k^2=1/(k+1), b_k^2=k/(k+1)$ for $k\in \mathbb{N}$ 
to get $v_k=\begin{pmatrix} a_k \\ b_k \end{pmatrix}$
and then the real idempotent 
$E_k= \begin{pmatrix} \frac{1}{k+1} & \frac{\sqrt{k}}{k+1}
    \\  \frac{\sqrt{k}}{k+1}
  & \frac{k}{k+1} \end{pmatrix}$. Now let $A_k=2E_k -I$ and then $A_k$ is a
unitary matrix. 

Then consider $\mathcal{W} = \{A_k | k\geq 1\}$. This is an infinite set
  of unitary matrices. Constellations may be formed  from 
  subsets of $\mathcal{W}$. The difference between any two of the
  unitary matrices in $\mathcal{V}$ is calculated as follows. 

\begin{proposition} Let $A_k,A_l \in \mathcal{V}$ where we assume $l >
  k$. Then $|\det(A_k-A_l)| = 4(\frac{\sqrt{l}-\sqrt{k}}{\sqrt{(l+1)(k+1)}})^2$.
\end{proposition}
\begin{proof} $|\det(A_k-A_l)| = |\det 2(E_k-E_l)| =
  4(\frac{\sqrt{l}-\sqrt{k}}{\sqrt{(l+1)(k+1)}})^2$ by Corollary
  \ref{detreal1}. 
\end{proof}
\begin{corollary}\label{constell1} $\frac{1}{2}|\det(A_k-A_l)|^{\frac{1}{2}}
  =\frac{\sqrt{k}-\sqrt{l}}{\sqrt{(l+1)(k+1)}}$.
\end{corollary} 
\begin{corollary}\label{constell2} $\mathcal{W} = \{A_k | k\geq 1\}$ is an infinite 
 constellation  with full diversity.
\end{corollary} 

Further
  now extend the $\mathcal{W}$ by taking the negatives of the elements
  which are also unitary matrices, that is let $\mathcal{X} = \{A_k,
  -A_k, | k\geq 1 \} = \mathcal{W}\cup -\mathcal{W}$.
\begin{proposition}\label{43} $|\det(A_k+A_l)|=
  4(\frac{\sqrt{l}+\sqrt{k}}{\sqrt{(l+1)(k+1)}})^2 $.\end{proposition}
\begin{proof} Note that $A_k+A_l= 2E_k-I + 2E_l-I = 2E_k -
  2(I-E_l)$. Now if $E_l$ is of the form $\begin{pmatrix} a^2 & ab
    \\ ba & b^2 \end{pmatrix}$ with $a^2 + b^2 = 1$ then $I-E_l
        = \begin{pmatrix} b^2 & -ab \\ -ba &
    a^2 \end{pmatrix}$.  Thus $I-E_l$ is the idempotent formed from
        $-b,a$ and apply Corollary \ref{detreal1} to get the result.  
\end{proof}  

Proposition \ref{modulus}  below may also be used to show that $|\det(A_k+A_l|
\geq |\det(A_k-A_l|$ which shows the quality is calculated from the
differences $|\det(A_k-A_l)|$.

\begin{corollary}\label{constell3} $\frac{1}{2}|\det(A_k+A_l)|^{\frac{1}{2}}=
  \frac{\sqrt{l}+\sqrt{k}}{\sqrt{(l+1)(k+1)}} $.
\end{corollary}
\begin{corollary}\label{constell4} $\mathcal{X} = \{A_k,
  -A_k, | k\geq 1 \} = \mathcal{W}\cup -\mathcal{W}$ is an infinite 
 constellation  with full diversity. The quality is the same as that
 of $\mathcal{W}$. 
\end{corollary} 
 
\quad 
Note that 
the distance of $A_k,A_l$ is smaller than the distance of $A_k,-A_l$.
  The quality may then be
determined by the least of these in the constellation chosen from
$\mathcal{W}$ or from $\mathcal{X}$. 

Finite constellations may be
chosen from $\mathcal{X}$ or $\mathcal{W}$ and the quality may be
directly calculated. This allows for a given rate $R$ 
the construction of $2\ti 2$ constellations
with full diversity and with this rate. Care should be taken in the
choices from $\mathcal{X}$ so as to ensure the quality is as large as
possible.   In Section \ref{tangle} this will be extended to $4\ti 4$,
$8\ti 8$ constellations and so on. 

\paragraph{Examples}
\begin{itemize} \item Consider $\mathcal{V} =
  \{A_1,A_2,A_3,A_4\} \cup \{-A_1,-A_2,-A_3,-A_4\}$. The rate here is
  $\log_28/2=1.5$. All the
  measured distances may be calculated and the least of these comes from
  $1/2|\det(A_3-A_4)|^{1/2} $ which is approximately $0.05991526$. To get better
  quality we need to take the unitary matrices to be `far
  enough apart'. 
\item\label{first3} Consider $\mathcal{V} =  \{A_1, A_2, A_4,
  A_{16}\}\cup \{-A_1, 
  -A_2, -A_4, -A_{16}\} $. The rate is
  again $1.5$ and the smallest measured distance comes from $A_4,A_2$
  giving the quality $\frac{\sqrt{16}-\sqrt{4}}{\sqrt{17*5}}$ which is
  approximately $ 0.217$.
 
\item Suppose it is required to construct a rate $R=2$ constellation
  of $2\ti 2$ matrices. Suppose also it is required that the
  constellation consist of real matrices. 
Then choose   $L=2^4=16$ unitary matrices from $\mathcal{X}$ or
$\mathcal{W}$. 

Consider    $\{A_1,A_2,\ldots, A_{8}, -A_1,-A_2,\ldots,
-A_{8}\}\}$. This has quality
$\frac{\sqrt{8}-\sqrt{7}}{\sqrt{9}*\sqrt{8}}$. This  is approximately 
    $0.02153$ but note the quality is  given in terms of elements of quadratic
extensions of $\mathbb{Q}$. 

We can do better by speading out the choice. 

By using roots of unity as in Section \ref{extend10} the quality may be
increased to approximately $0.3826$ while still maintaining
non-commutativity.
 
\end{itemize}

\paragraph{More generally:}

For  $p/q \in \mathbb{Q}$ with $p,q \in \mathbb{N},\, p< q$ 
define $a_{p,q}^2 = p/q, b_{p,q}^2 = (q-p)/q$ and $v_{p,q}
= \begin{pmatrix} a_{p,q} \\ b_{p,q} \end{pmatrix}$ to form the
idempotent $E_{p,q}=\begin{pmatrix} \frac{p}{q} & \frac{\sqrt{p(q-p)}}{q} 
\\ \frac{\sqrt{p(q-p)}}{q} & \frac{q-p}{q}\end{pmatrix}$. Then define
$A_{p,q} = 2 E_{p,q}-I$ which is then a unitary matrix. 

For example
$E_{4,7} = \begin{pmatrix} \frac{4}{7} & \frac{\sqrt{12}}{7}
  \\ \frac{\sqrt{12}}{7} & \frac{3}{7} \end{pmatrix}$. 

Constellations may then be formed from $\{A_{p,q}, -A_{p,q}\}$ by varying
$p,q \in \mathbb{N}, \, q>p$. For example let  $p=1$, always, and vary
$q$ so as to give the previous system of constellations. Another
example is let $p=2$ and vary $q$ with $q>2$. 
The distances are relatively easy to work out. Which of these types 
are best needs to be worked on.
   
\subsection{Real $\cos,\sin$}\label{sincos}
Take now $a=\cos \theta, b=\sin \theta$ to form $E=\begin{pmatrix}
  \cos^2\theta & \cos\theta \sin \theta \\ \cos \theta\sin \theta &
  \sin^2\theta \end{pmatrix} =  \begin{pmatrix}
  \cos^2\theta & \sin 2\theta/2 \\ \sin 2\theta/2 &
  \sin^2\theta \end{pmatrix}$. Then let $A=2E-I$.
 
Similarly form $F=\begin{pmatrix}
  \cos^2\al & \cos\al \sin \al \\ \cos \al\sin \al &
  \sin^2\al \end{pmatrix} =  \begin{pmatrix}
  \cos^2\al & \sin 2\al/2 \\ \sin 2\al/2 &
  \sin^2\al  \end{pmatrix} $. 
Then let $B=2F-I$.

Then $|\det(A-B)|=|\det2(E-F)| = 4(\cos \theta \sin \alpha - \sin \theta\
\cos \alpha)^2 = 4\sin^2(\alpha-\theta)$.

Build the constellation from such unitary matrices. Thus it is
required to build the constellation  so  that $|\sin
(\theta-\al)|$ is as large as possible for all  $\theta, \al$ used in
forming the constellation.

Define $E_{j,0} = \begin{pmatrix} \cos^2 \theta_j
& \sin(\theta_j)\cos
(\theta_j) \\ \cos(\theta_j)\sin(\theta_j)&\sin^2(\theta_j)\end{pmatrix}$
and $E_{j,1}= I-E_{j,0}$. Then let $U_j=E_{j,0}-E_{j,1}= 2E_{j,0}-I$
which is a unitary matrix.  

Let the constellation then be defined as follows:
Define $\theta_j=2j\pi/n$ for $j=0,1,2,\ldots, n-1$ and
$U_j=2E_{j,0}-I$ be as above and the constellation is defined by  
$\mathcal{V}_n = \{U_0,U_1, \ldots, U_{n-1}\}$. 
\begin{proposition}\label{odd} For odd $n$ the quality of the constellation
  $\mathcal{V}_n$ is $|\sin(\pi/n)|$.
\end{proposition}
\begin{proof} By Corollary \ref{detreal1}, for $A,B \in \mathcal{V}_n,\, A\neq B,
|\det(A-B)| = 4|\sin(\theta_j-\theta_k)|^2$   
for $j\neq k, 0\leq j,k \leq n-1$. Thus  $|\det(A-B)|= 
4|\sin (2r\pi/n)|^2$ where $r = j-k$. This is never $0$ for odd $n$ and
then it has minimum value $ 4|\sin (\frac{(n+1)\pi}{n})|^2=4|\sin(\pi/n)|^2$ 
attained when $r=(n+1)/2$, that is when $j=(n+1)/2, k=0$ or
$j=(n+3)/2, k=1$ and others. Note that $|\sin(\pi +\al)| = |\sin \al|$
and that for $\al < \be < \pi/2$ that $|\sin \be| > |\sin\al|$.    

The quality of the constellation (for odd $n$) is thus
$\zeta_n=\frac{1}{2}(4|\sin^2(\pi/n)|^{\frac{1}{2}}=
|\sin(\pi/n)|$ .
\end{proof}  

\vspace{.1in}
Thus for example:

$n=5$: Rate $R=\log_25/2=1.1609..$ , $\zeta_5
= (\sin(\pi/5))=0.58778....$. Of the 10 possible differences, 5 have
difference $0.58778..$ and 5  have difference $0.95105..$ 

$n=9$. Rate $R=\log_29/2=1.5849..$, $\zeta_9
= (\sin(\pi/9))= 0.3420....$. 

$n=17$. Rate $R=\log_217/2 = 2.0437 ...$, and quality 
$\zeta_{17} = \sin(\pi/17)= 0.1837..$


The constellation $\mathcal{V}_n$ consists of  real  unitary matrices.

\subsection{Extend the range}\label{improve} 
Consider the $U_i$ constructed in Section \ref{sincos}.
Note that $-U_j = 2(I-E_{0,j})$ is also a unitary matrix and we now
include these with the  constellation already constructed. Thus consider 
$\mathcal{X} = \mathcal{V}_n \cup -\mathcal{V}_n = \{U_0,U_1, \ldots,
U_{n-1}\} \cup  \{-U_0,-U_1, \ldots, - U_{n-1}\}$.
 
$U_j-(-U_j)= 2(U_j)$. Now $|\det(2(U_j)|= 4$. Then the difference of
$U_j, -U_j$ is 
$\frac{1}{2}|\det(2(U_j)|^{1/2}= 1$.

As with Proposition \ref{43} it is shown that $|\det(U_i+U_j)| >
|\det(U_i-U_j|$ so that the quality of the constellation is determined
by the differences $|\det(U_i-U_j)|$. 
 
\begin{proposition} For $n$ odd the quality of $\mathcal{X}$ is 
$|\sin(\pi/n)|$. 
\end{proposition}
\begin{proof} This follows from Proposition \ref{odd} since $|\det(U_i+U_j)| >
|\det(U_i-U_j|$.
\end{proof}

The rate is better here. 

So for example when $n=5$ get quality $0.58778..$ and rate $\log_2(10)/2=1.6609..$.\\
When $n=9$ get quality $\sin(\pi/5)=0.3420..$ and rate $\log_2(18)/2=2.0849..$.\\
When $n=17$ get quality $\sin(\pi/17)= 0.183749..$ and rate
$\log_2(34)/2=2.5437.. $.

 To get constellations with real entries of higher degree it is
 necessary to use constellations constructed from different sets of
 complete idempotents. Alternatively the $2\ti 2$ case can be
 upper loaded to $4 \ti 4$ and then to $8\ti 8$ as shown in 
 Section \ref{tangle}.

\subsection{Range further}\label{extend10}

Suppose now a constellation $\mathcal{V}= \{A_1,A_2, \ldots, A_t\}$ has been
built. Assume the elements are real unitary matrices although this is
not necessary in general.

We can extend the range as follows. For each $j$ let  $\mathcal{V}_j= 
\{A_j, \om A_j, \om^2 A_j, \ldots, \om^{k-1}A_j\}$ where $\om$ is a
primitive $k^{th}$ root of unity and define $\mathcal{V}_\om =
\cup_{j=1}^{k} \mathcal{V}_j$. (The case $k=2, \om =
-1$ was considered in Section \ref{improve}.)

Then the following may be shown for unitary matrices constructed from
the idempotents of type  $\begin{pmatrix} aa^* & ab^* \\ ba^* & bb^* \end{pmatrix}$. 

\begin{proposition}\label{porp} The quality of $\mathcal{V}_\om$ is
 $\min\{\sin(\pi/k), \zeta\}$ where $\zeta$ is the quality of
 $\mathcal{V}$.
\end{proposition}
\begin{proof} Since $|\det(\om^kA-\om^jB)| = |\det(A-\om^{j-k}B)|$ as
  $|\om|=1 $ it is only necessary to look at differences
  $|\det(A-\om^kB)|$. Now $|\det(A-\om^j A)|  = |\det (1-\om^j)A| =
 |(1-\om^j)^2 \det A| = |(1-\om^j)^2|$ since $A$ is unitary. Now
 $|(1-\om^j)^2|$ is least when $j=1$. We are interested then in the
 difference 
 $1/2{(|1-\om|^2})^{\frac{1}{2}}= 1/2|1-\om|$. Now $1/2|1-\om| =
 1/2|1-e^{2\pi/k}|= 1/2|\sqrt{2-2\cos 2\pi/k}|=1/2|\sqrt
 4\sin^2\pi/k|=|\sin \pi/k|$. 

Now consider $|\det(A-\om^jB|$. By Proposition \ref{gone} this is $|\om\det(A-B)| =
 |\det(A-B)|$. Thus $\mathcal{V}_\om = \min\{\sin(\pi/k),\zeta\}$. 

\end{proof}

Now apply this result to get examples of constellations as follows.

\begin{enumerate} \item\label{junk} Consider the constellation 
$\{B_1=A_1, B_2 = A_2, B_3=A_4, B_4=A_{16}\}$ as in
 Section \ref{first2}.  As was shown in that Section, this has quality $\zeta =
		      \frac{\sqrt{16}-\sqrt{4}}{\sqrt{85}}$ which is
		      approximately $0.217$. 
Now extend this to the constellation \\ $\mathcal{V}=\{B_1,B_2,B_3,B_4, \om
		      B_1,\om B_2, \om B_3, \om B_4, \om^2B_1,
		      \om^2B_2,\om^2B_3,\om^2B_4,
		      \om^3B_1,\om^3B_2,\om^3B_3,\om^3B_4\}$ where $\om$
		      is a primitive $4^{th}$ of unity. This is a
		      constellation of sixteen $2\ti 2$ matrices which
		      by  Proposition \ref{porp} has quality
		      $\min\{\zeta, \sin (\pi/4)\} = \zeta \approx 0.217$. The rate is
		      $\log_2{16}/2=2$. 
\item This is the similar to previous example \ref{junk}.\ except now
  let $\om$ be an
$8^{th}$ root of unity and consider the constellation of 32 unitary
matrices obtained in this way. The rate is $\log_2(32)/2=2.5$ and the
quality is $\min\{\zeta,\sin{\pi/8}\}= \zeta$ as $\sin(\pi/8) \approx
0.38268..$.
\item Let now $\om$ be a primitive $16^{th}$ root of unity and then as
  in example \ref{junk}.\ get a constellation of $64$ unitary matrices with
  rate $\log_2(64)2=3$ and quality $\min\{\zeta, \sin(\pi/16)\}=
  \sin(\pi/16) \approx 0.1950$.
\end{enumerate} 

The above are just a few of the constellations that may be
constructed by this method. 

Once $2\ti 2$ constellations are constructed they  may  be used 
to get $4\ti 4$,
then $8\ti 8$  constellations etc.\ by the methods of Section
\ref{tangle} and these new ones will have similar quality and rate.
\subsection{Constellations constructed from complex symmetric
  idempotents}\label{complex} In Sections \ref{first2} and \ref{sincos}, real unitary
matrices are constructed from orthogonal sets of idempotents. Now
complex constellations are constructed from complex
idempotents. 

We specialise to $2\ti 2$ unitary matrices formed from complex
symmetric orthogonal idempotents although other cases may also be
considered.

We construct $2\ti 2$ unitary matrices from the idempotents $uu^*$
where $u$ is a unit vector and may have complex entries.  
Then as noted $E=uu^*
= \begin{pmatrix} aa^* & ab^* \\ ba^* & bb^* \end{pmatrix}$ and
$aa^*+bb^* = 1$. 

The unitary matrix $2E-I$ is formed. The differences between such
unitary may be calculated from Corollary \ref{detreal1}.

Let $v= \begin{pmatrix} a \\ b \end{pmatrix}$  be a unit vector in $\cc_2$. 
Form $E_{a,b}=vv^*$ which is then an idempotent
and let $A_{a,b}=2E_{a,b}-1$ which is a unitary matrix.

Constellations are then formed from such unitary matrices and the
differences are calculated from Corollary \ref{detreal1}.
This is a very general construction and there is no restriction on $v
\in \cc_2$ except that it be a unit vector.

Now consider constellations which will have  entries in $\mathbb{Q}(i)$. 
Consider a vector $u= \begin{pmatrix} a \\ b\end{pmatrix}$ with $a,b 
\in \mathbb{Z}(i)$ where  $|u| = \sqrt{t}$ with
$t\in \mathbb{N}$. Then $v=\frac{1}{\sqrt{t}}u$ is a unit vector  and the
idempotent formed from $v$ has the form $E=\frac{1}{t}\begin{pmatrix} 
aa^* & ab^* \\ ba^* & bb^* \end{pmatrix}$  and the corresponding
unitary matrix has the form $2E-I$ which has entries in
$\mathbb{Q}(i)$. We refer to this matrix as $A_{a,b}$ although 
$(c,d)=\al(a,b)$ will produce the same unitary matrix. 

Form constellations of the form $\{A_{a_j,b_j} | j\in J, a_j,b_j \in
\mathbb{Z}, \, (a_k,b_k)\neq \al(a_j,b_j) k\neq j\}$.

For example: \\ (1) $a=1+2i, b=2+i$, $E_1=\frac{1}{10}\begin{pmatrix} 5 &
  4+3i \\ 4-3i & 5\end{pmatrix}$, (2)
$a=1+3i,b=3+i$, $E_2=\frac{1}{20}\begin{pmatrix} 10 &
  6+8i \\ 6-8i & 5\end{pmatrix}$, \\ (3) $a=2+3i, b= 3+i$,
  $E_3=\frac{1}{26}\begin{pmatrix} 13 & 
  12+5i \\ 12-5i & 5\end{pmatrix}$, (4): $a=2+3i,b=1+i$,
  $E_4=\frac{1}{15}\begin{pmatrix} 13 & 5+i
   \\ 5-i & 2\end{pmatrix}$.

Let $A_i=2E_i-I$ and form the constellation $\{A_1,A_2,A_3,A_4,-A_1,
-A_2,-A_3,-A_4\}$. The quality of the constellation is easily worked
out using Corollary \ref{detreal1} and is left as an exercise.

\subsection{Extending  the range} Here it is shown how to extend 
 constellations which may be done without loss of quality. 
\begin{lemma}\label{lemma4} Let $A$ be a unitary $m\ti m$ matrix and $\om = e^{i\theta}$ a
 complex number of modulus 1. Then $\frac{1}{2}|\det(A-\om A)|^{\frac{1}{m}}=
   |\sin (\frac{\theta}{2})|$.
\end{lemma}
\begin{proof} $\frac{1}{2}|\det(A-\om
  A)|^{\frac{1}{m}}=\frac{1}{2}|\det((1-\om)A)|^{\frac{1}{m}}=
  \frac{1}{2}|(1-\om)^m\det(A)|^{\frac{1}{m}}=
  \frac{1}{2}|(1-\om)^m|^{\frac{1}{m}} = 
  \frac{1}{2}|1-\om|= \frac{1}{2}| \sqrt{(1-\om)(1-\om^*)}| = \frac{1}{2} |
  \sqrt{2-(\om+\om^*)}|= \frac{1}{2}|\sqrt{2(1-\cos
    (\theta))}=\frac{1}{2}|\sqrt{2*2\sin^2(\frac{\theta}{2})}|=
    |\sin(\frac{\theta}{2})|$.
\end{proof}
\begin{proposition}\label{modulus} Let $A,B$ be unitary matrices and $\om=e^{i\theta}$
  a complex number of modulus $1$. Then $|\det(A-\om B)|\geq
  |\det(A-B)|$.
\end{proposition}
\begin{proof} Now  $|\det(A-\om B)| = |\det (A(I-\om A^*B))| =
  |\det(A)\det(I-\om A^*B)|= |\det(A)||\det(I-\om A^*B)| = |\det(I-\om A^*B)|$
  as $A$ is unitary. Similarly $|\det(A-B)|=|\det(I-A^*B)|$.

Thus it is only necessary to show $|\det(I-\om X)| \geq |\det(I-X)|$
for a unitary matrix $X$. 

Let $k$ be the size of the matrices. 
As $X$ is unitary these exists a unitary
matrix $P$ such that $P^*XP=D$ where $D$ is a diagonal matrix with
diagonal entries $\{d_1,d_2,\ldots, d_k\}$. Then also $P^*(I-X)P =
D_0$ where $D_0$ is diagonal with diagonal entries $\{1-d_1, 1-d_2, \ldots, 1-d_k\}$
and $P^*(I-\om X) =D_1$ where $D_1 $ is diagonal with diagonal entries $\{1-\om d_1,
1-\om d_2, \ldots, 1-\om d_k\}$. 

Then $|\det(I-X)| = |\prod_{i=1}^k(1-d_i)|= \prod_{i=1}^k |1-d_i|$ and   
 $|\det(I-\om X)| = |\prod_{i=1}^k(1-\om d_i)|= \prod_{i=1}^k |1-\om
d_i|$. 

For complex numbers $|z_1-z_2|\geq ||z_1|-|z_2||$. Let $z_1=1, z_2=
\om d$ and then $|1-\om d| \geq ||1| - |\om d|| = |1- d|$ as $|\om|
=1$. Thus $|1-\om d_i| \geq |1-d_i|$ for each $i$ and so $|\det(I-\om
X)| \geq |\det(I-X)|$.

\end{proof}

\begin{proposition}\label{prop2} 
Let $\mathcal{V}= \{A_1,A_2,\ldots, A_n\}$ be a fully diverse 
constellation of $n$
matrices with quality $\zeta$ and $\om= e^{2\pi i/k}$ a primitive $k^{th}$ root of
unity. Define $\mathcal{V}_{i, \om}=\{A_i, \om A_i, \ldots, \om^{k-1}
A_i\}$ for $i=1,2,\ldots, n$ and $\mathcal{V}_\om = \cup_{i=1}^n
\mathcal{V}_{i,\om}$. Suppose $A_j\neq \om^tA_l$ for $j\neq l$ and
for $1\leq t \leq k-1$. Then
the quality of $\mathcal{V}_\om$ is $\min\{\zeta,
|\sin(\frac{\pi}{k})|\}$.
\end{proposition}
\begin{proof}
The result follows from  Proposition \ref{modulus} and Lemma
\ref{lemma4}. Note that $|\det(\om^kA-\om^jB)|=|\det(A-\om^{j-k}B)|$
and that $|\sin(\frac{r\pi}{k})| \geq |\sin(\frac{\pi}{k})$ for $1\leq
r\leq k-1$.
\end{proof} 

This enables the construction of a constellation with $kn$ elements
from a constellation $\mathcal{V}$ with $n$ elements and the quality
is the same provided the quality of  $\mathcal{V}$ is greater than or
equal to $\sin(\pi/k)$.   

\section{Tangle to construct higher order constellations}\label{tangle}

In \cite{hur6} and \cite{hur5} the idea of a {\em tangle of matrices} is introduced. This
construction is now used to construct constellations of matrices of
higher order  from constellations of smaller order matrices.

Suppose $A,B$ are matrices of the same size. Then a {\em tangle} of
$\{A,B\}$ is  one of
\begin{enumerate}
\item\label{one} $W=\frac{1}{\sqrt{2}}\begin{pmatrix} A& A \\ B &-B
				    \end{pmatrix}$.

\item\label{two} $W=\frac{1}{\sqrt{2}}\begin{pmatrix} A& B \\ A &-B
				    \end{pmatrix}$.

Note that \ref{two}. is the transpose of \ref{one}.

\end{enumerate}

A tangle of $\{A,B\}$ is not the same as,  and is  
 not necessarily equivalent to, a tangle of $\{B,A\}$ which is one of 
$\frac{1}{\sqrt{2}}\begin{pmatrix} B& A \\ B &-A
				    \end{pmatrix},
\frac{1}{\sqrt{2}}\begin{pmatrix} B& -A \\ B &A
				    \end{pmatrix}$.

Note that interchanging any rows and/or columns of a unitary matrix
results in a unitary matrix. 

If $A=B$ then a tangle of $\{A,A\}$ is a tensor product 
 but a tangle of $\{A,B\}$ is not necessarily a tensor product when
$A\neq B$; this is why they can can be useful for constructions. 

Then as in \cite{hur5} or \cite{hur6} the following may be shown.\footnote{The result in \cite{hur5} and \cite{hur6} is more general
   where it  is shown  for {\em paraunitary matrices}; unitary
   matrices are special cases of paraunitary matrices. }
\begin{proposition}\label{fty} Let $A,B$ be unitary matrices of the
 same size.  Then a tangle of $\{A,B\}$ or of $\{B,A\}$ is a
 unitary matrix.
\end{proposition}
\begin{proof} This is shown for $W=\frac{1}{\sqrt{2}}\begin{pmatrix}
    A& A \\ B &-B\end{pmatrix}$; the proofs for the others are similar. 
Now $WW^* =\frac{1}{\sqrt{2}}\begin{pmatrix} A& A \\ B
  &-B \end{pmatrix}\frac{1}{\sqrt{2}}\begin{pmatrix} A^*& B^* \\ A* &-B^*
\end{pmatrix} = \frac{1}{2}\begin{pmatrix}AA^* + AA^* & AB^* - AB^*
  \\  BA^* -BA^* & BB^* + BB^*\end{pmatrix} =
  \frac{1}{2}\begin{pmatrix} 2I_n & \un{0}
  \\ \un{0} & 2I_n\end{pmatrix} = I_{2n}$.
\end{proof}

\begin{lemma}\label{dett}
 Let $A$ be an $n\ti n$ matrix. Then $\det(\al A) = \al^n\det(A)$
 for a scalar $\al$.
\end{lemma}

Given a constellation of size $n\ti n$ there are a number of ways of
constructing constellations of size $2n\ti 2n$ from a constellation of
size $n\ti n$ using tangled products. 


Here is an example to explain the method in general. 
Let $\mathcal{V}_0=\{A_1,A_2,A_3,A_4, -A_1,-A_2, -A_3, -A_4\}$ be a constellation of
$m\ti m$ unitary matrices as for example constructed in Sections
\ref{improve} or \ref{complex}.
 
Then consider the following constellation of $2m\ti 2m $ matrices.

$\mathcal{V}= \frac{1}{\sqrt{2}}\begin{pmatrix} A_1 & A_1 \\ A_1 &
  -A_1 \end{pmatrix},\frac{1}{\sqrt{2}}\begin{pmatrix} A_2 & A_2 \\ A_2 &
  -A_2 \end{pmatrix},\frac{1}{\sqrt{2}}\begin{pmatrix} A_3 & A_3 \\ A_3 &
  -A_3 \end{pmatrix},\frac{1}{\sqrt{2}}\begin{pmatrix} A_4 & A_4 \\ A_4 &
  -A_4 \end{pmatrix}, \\ \frac{1}{\sqrt{2}}\begin{pmatrix} -A_1 & -A_1 \\ -A_1 &
  A_1 \end{pmatrix},\frac{1}{\sqrt{2}}\begin{pmatrix} -A_2 & -A_2 \\ -A_2 &
  A_2 \end{pmatrix},\frac{1}{\sqrt{2}}\begin{pmatrix} -A_3 & -A_3 \\ -A_3 &
  A_3 \end{pmatrix},\frac{1}{\sqrt{2}}\begin{pmatrix} -A_4 & -A_4 \\ -A_4 &
  A_4 \end{pmatrix}, 
\\ \frac{1}{\sqrt{2}}\begin{pmatrix} A_1 & -A_1 \\ A_2 &
  A_2 \end{pmatrix},\frac{1}{\sqrt{2}}\begin{pmatrix} A_2 & -A_2 \\ A_1 &
  A_1 \end{pmatrix},\frac{1}{\sqrt{2}}\begin{pmatrix} A_3 & -A_3 \\ A_4 &
  A_4 \end{pmatrix},\frac{1}{\sqrt{2}}\begin{pmatrix} A_4 & -A_4 \\ A_3 &
  A_3 \end{pmatrix}, \\ \frac{1}{\sqrt{2}}\begin{pmatrix} -A_1 & A_1 \\ -A_2 &
  -A_2 \end{pmatrix},\frac{1}{\sqrt{2}}\begin{pmatrix} -A_2 & A_2 \\ -A_1 &
 - A_1\end{pmatrix},\frac{1}{\sqrt{2}}\begin{pmatrix} -A_3 & A_3 \\ -A_4 &
  -A_4 \end{pmatrix},\frac{1}{\sqrt{2}}\begin{pmatrix} -A_4 & A_4 \\ -A_3 &
  -A_3 \end{pmatrix} $

The first 8 could be considered as tensor products. Notice that the
second 4 are the negatives of the first four but they could also be
considered as tangles of the negatives $-A_1,-A_2,-A_3,-A_4$.

\begin{proposition}\label{prop3} Suppose $\mathcal{V}_0$ has quality $\zeta$. Then
  the quality of $\mathcal{V}$ is $2^{\frac{1}{4}}\zeta$.
\end{proposition}
\begin{proof} The proof consists of working out the differences.
We show the proof when the blocks are of size $2\ti 2$; the other
 cases are similar.

Consider the difference of the matrices 

$|\det(\frac{1}{\sqrt{2}}\begin{pmatrix} A_1 & A_1 \\ A_1 &
  -A_1 \end{pmatrix} - \frac{1}{\sqrt{2}}\begin{pmatrix} A_2 & A_2 \\ A_2 &
  -A_2 \end{pmatrix}| \\ =|\det( \frac{1}{\sqrt{2}} \begin{pmatrix}
A_1-A_2 & A_1-A_2 \\ A_1-A_2
  & -A_1+A_2 \end{pmatrix})| \\ = |\det(\frac{1}{\sqrt{2}}
 (A_1-A_2)(-A_1+A_2)- 
(A_1-A_2)(A_1-A_2))|$ 

 since $(A_1-A_2)$ and
 $(-A_1+A_2)$ commute.

Thus this difference $\delta =  |\det(\frac{1}{\sqrt{2}}2(A_1-A_2)^2| =
 |\det(\frac{2}{\sqrt{2}}(A_1-A_2)^2)| = |\frac{4}{2}(\det(A_1-A_2)^2)|= 
 2|\det(A_1-A_2)^2)|$, by Lemma \ref{dett} as the $A_i$ are $2\ti 2$
 matrices. 

Now it is known that $\frac{1}{2}|\det(A_1-A_2)|^{1/2}\geq \zeta$ and so
$\frac{1}{2}\delta^{\frac{1}{4}} = \frac{1}{2}(2|\det(A_1-A_2)^2|)^{1/4})
 = \frac{1}{2}
 2^{\frac{1}{4}}|\det(A_1-A_2)|^{\frac{1}{2}}\geq 2^{\frac{1}{4}}\zeta$.

Similarly the other differences are shown to be $\geq 2^{\frac{1}{4}}\zeta$. 

Note that $d^{\frac{1}{2}} \geq d$ when $0\leq d \leq 1$; this is needed for
 some of the other difference calculations.

It
 is clear also, since $\mathcal{V}_0$ contains the exact difference
 $\zeta$,  that
 the difference $2^{1/4}\zeta$ is attained by $\mathcal{V}$.  

\end{proof}

Thus in this manner  it is possible to start out with a constellation of eight
$2\ti 2$ matrices of 
quality $\zeta$, then construct a constellation of sixteen 
 $4\ti 4$ matrices of quality $2^{1/4}\zeta$ from these 
 construct a constellation of thirty-two  $8\ti
8$ matrices of quality $2^{1/4}2^{1/4}\zeta$ and so on. The rates go
from $\log_28/2=1.5$ to $\log_216/4= 1$ to $\log_232/8=0.625$ with
higher order and the quality goes up slightly. By
starting out with 16 in the original constellation 
the rates go from $2$ to $1.25$ to $.75$ with higher order. 

\subsection{A general construction using tangles}
More generally proceed as follows. Consider a constellation $\mathcal{V}_0=
\{A_1, A_2, A_3, \ldots, A_k\}$ of $n\ti n$ matrices such that 
$A_i\neq \om^s A_j$ for $i\neq j$ where $\om $ is a primitive $t^{th}$ root of
unity. Consider $k=2w$ to make the  explanation slightly simpler  
but this isn't necessary. Let $B_i=\frac{1}{2}\begin{pmatrix} A_i & A_i \\ A_i &
-A_i \end{pmatrix}$ for $i=1,2,\ldots, 2w$ 
and then for $i=1,3, \ldots, 2w-1$ define
$C_i= \frac{1}{2}\begin{pmatrix} A_i & -A_{i} \\ A_{i+1} &
  A_{i+1}\end{pmatrix}, D_i =\frac{1}{2}\begin{pmatrix} A_{i+1}
  &-A_{i+1} \\ A_i & A_i \end{pmatrix}$. Let 
$\mathcal{V}_i = \{B_i, \om B_i, \om^2B_i, \ldots, \om^{t-1}B_i\}$ for
$i=1, \ldots, 2w$ and $\mathcal{W}_i = \{C_i, \om C_i, \ldots,
\om^{t-1}C_i, D_i, \om D_i , \ldots, \om^{t-1}D_i\}$ for
$i=1,3,\ldots, 2w-1$.

Let $\mathcal{V}_{0,\om} = \cup_{i=1}^{2w}\mathcal{V}_i
\cup_{k=0}^{w-1}\mathcal{W}_{2k+1}$. Now $\mathcal{V}_{0,\om}$ consists of
$2n\ti 2 n$ unitary matrices and has $4kt$ elements. 

\begin{proposition} The quality of $\mathcal{V}_{0,\om}$ is
  $\min\{2^{1/4}\zeta, \sin(\pi/t) \}$
  where $\zeta$ is the quality of $\mathcal{V}_0$.
\end{proposition}

The proof is omitted but depends on Propositions \ref{modulus} and
\ref{prop2} and calculations similar to those in  Proposition \ref{prop3}. 

\paragraph{Samples}

Let $\mathcal{V}_0=\{A_1,A_2, A_4, A_{16}\}$ be a constellation of $2\ti 2$
matrices with 
quality $\zeta\approx 0.217$, as in Section \ref{first2},  and $\om$ a
primitive $8^{th}$  root of unity. Then 
$\mathcal{V}_{0,\om}$ has quality  $\min\{2^{1/4}\zeta, \sin(\pi/8) \}$
and rate $\log_264/4= 3$. Now $\sin(\pi/8) \approx 0.3827$ so
$2^{1/4}\zeta \approx 0.258$ is the quality.

Consider the elements in $\mathcal{V}_{0,\om}$ not involving $\om$ and
these form a constellation $\mathcal{W}_0$ of $8$ unitary $4\ti 4$
matrices. Now form $\mathcal{W}_{0,\om}$ where $\om$ is again a
primitive $8^{th}$ root of unity. This gives a constellation of 128
unitary $8 \ti 8$ matrices which has quality  $\min\{2^{1/4}2^{1/4}\zeta,
\sin(\pi/8) \}$. Now $2^{1/2}\zeta \approx 0.3069$ and this is the
quality. The rate is $\log_2(128)/8=  0.875$. 

\subsection*{General conclusion} 
The methods allow the construction of constellations of many types
and sizes and the quality may be calculated directly. In many cases
the complete set of differences can be worked out as
required. Idempotents are building blocks for unitary matrices. The
samples given within are a small subset of the possibilities  and many
more fully diverse constellations may be developed by the methods. 


\begin{thebibliography}{99}
\bibitem{mult} A. Shokrollahi, B.  Hassibi, B.M.   Hochwald,
  W. Sweldens, ``Representation
	theory for high-rate multiple-antenna code design'', IEEE
	Trans. on Inform. Theory, 47, no.6, (2001), 2335-2367.
\bibitem{berhuy}

Gr\'egory Berhuy and Fr\'ed\'erique Oggier, {\em An Introduction to Central
	Simple Algebras and Their Applications to Wireless
	Communications,} AMS Mathematical Surveys and Monographs, Vol
	191, Providence RI, 2013. 
\bibitem{channel} {\em Channel Coding: Theory, Algorithms, and
	Applications},
 Edited by: David Declerq, Marc Fossorier and Ezio Biglieri, Academic
	Press Library in Mobile and Wireless Communications, Chapter 10,
	2014. 

\bibitem{orig} B. Hochwald, W. Sweldens, ``Differential unitary space
  time modulation'', IEEE Trans. Comm., 48, (2000), 2041-2052.
\bibitem{ams} B.A. Sethuraman, 
``Division Algebras and Wireless Communication'', Notices of the AMS,
  57, no.\ 11 (2010), 1432-1439. 
\bibitem{blahut} Richard E.\ Blahut, 
{\em Algebraic Codes for data transmission}, Cambridge University
Press, 2003. 
\bibitem{commblock} I. Kovacs, D.S. Silver, S.G. Williams,
  ``Determinants of Commuting-Block matrices'', Amer.\ Math.\ Monthly,
  106, no.\ 10, 950-952, 1999. 
\bibitem{hur1} Paul Hurley and Ted Hurley, ``Codes from zero-divisors
   and units   in group rings'', Int. J. Inform. and Coding Theory, 1
	(2009), 57-87.
\bibitem{hur6}  Barry Hurley and Ted Hurley, ``Paraunitary matrices
  and group rings'', Intn. J. Group Theory, 3, no.1, 31-56, 2014.
\bibitem{hur5} Barry Hurley and Ted Hurley, ``Paraunitary matrices'',
  arXiv:1205.0703v1.
\bibitem{idemrank} Oskar M. Baksalary, Dennis S. Bernstein,
	Götz Trenkler, ``On the equality between rank and trace of an
	idempotent matrix'', Applied Mathematics and Computation, 217,
	4076-4080, 2010. 

\bibitem{seh} C\'esar Milies \& Sudarshan Sehgal, {\em An introduction to
  Group Rings}, Klumer, 2002.
\end{thebibliography}
\end{document}